\newif\ifconf
\conftrue %last one wins
\conffalse

\newif\ifnotes
\notestrue %last one wins
\notesfalse

\documentclass[letterpaper,11pt]{article}

\usepackage{typearea}
\typearea{15}
\usepackage[compact]{titlesec}

\newcommand{\confoption}[2]
{\ifconf%
#1%
\else%
#2%
\fi}

\newenvironment{packed_item}{
\begin{itemize}
  \setlength{\itemsep}{0.5pt}
  \setlength{\parskip}{0pt}
  \setlength{\parsep}{0pt}
}{\end{itemize}}

\usepackage[margin=1in]{geometry}
\usepackage{appendix}
\usepackage{amsthm,amsfonts,amsmath,amssymb,dsfont}
\usepackage{cite,url,balance}
\usepackage{algorithm,algorithmic}
\usepackage{graphicx,psfrag}
\usepackage[usenames,dvipsnames]{color}
\usepackage{array,multirow}
\usepackage{xspace}

\newcounter{note}[section]
\renewcommand{\thenote}{\thesection.\arabic{note}}
\ifnotes
\newcommand{\annote}[1]{\refstepcounter{note}$\ll${\bf Sasho~\thenote:} {\sf #1}$\gg$\marginpar{\tiny\bf AN~\thenote}}
\newcommand{\ktnote}[1]{\refstepcounter{note}$\ll${\bf Kunal~\thenote:} {\sf #1}$\gg$\marginpar{\tiny\bf KT~\thenote}}
\else
\newcommand{\annote}[1]{}
\newcommand{\ktnote}[1]{}
\fi
\newtheorem{theorem}{Theorem}[section]
\newtheorem{definition}[theorem]{Definition}

\newtheorem{lemma}[theorem]{Lemma}
\newtheorem{corollary}[theorem]{Corollary}
\newtheorem{cor}[theorem]{Corollary}

\numberwithin{algorithm}{section}

\def\tr{{\rm tr}} % trace
\def\E{\mathbb{E}\ } % expectation
\def\Pr{{\rm Pr}} % probability
\def\R{{\mathbb{R}}} % set of real number
\newcommand{\alg}{\mathcal{A}}
\newcommand{\mech}{\mathcal{M}}
\newcommand{\eps}{\varepsilon}
\renewcommand{\vec}[1]{\mathbf{#1}}

\newcommand{\junk}[1]{}

\def\b0{{\bf 0}}

\DeclareMathOperator{\conv}{conv}

\DeclareMathOperator{\diag}{diag}
\DeclareMathOperator{\supp}{supp}

\def\poly{\operatorname{poly}}
\DeclareMathOperator{\diam}{diam}

\DeclareMathOperator{\marg}{marg}
\DeclareMathOperator{\parity}{par}
\DeclareMathOperator{\err}{err}

\newcommand{\cut}[1]{}

\newcommand{\afive}{\textsc{Projection}}
\newcommand{\afrelax}{\textsc{RelaxedProj}}
\newcommand{\fw}{\textsc{FrankWolfe}}

\newcommand{\univ}{\mathcal{U}}
\newcommand{\queries}{\mathcal{Q}}

\newcommand{\one}{\mathbf{1}}
\newcommand{\etal}{{\em et al.}}
\newcommand{\dee}{p}
\newcommand{\ex}{\mathbb{E}}
\newcommand{\maxkxor}{{\sc MAX$k$-XOR}\xspace}

\begin{document}
\title{Efficient Algorithms for Privately Releasing Marginals via Convex Relaxations}
\author{Cynthia Dwork\thanks{Microsoft Research SVC, Mountain View, CA 94043.}
 \and Aleksandar Nikolov\thanks{Rutgers University, Piscataway, NJ 08854.} \and Kunal Talwar\thanks{Microsoft Research SVC, Mountain View, CA 94043.}}
\date{}
\maketitle
\vspace{-30pt}
\begin{center}{\large{[Full Version]}}\end{center}
\vspace{10pt}
\thispagestyle{empty}
\setcounter{page}{0}

\begin{abstract}
Consider a database of $n$ people, each represented by a bit-string of length $d$ corresponding to the setting of $d$ binary attributes. A $k$-way marginal query is specified by a subset $S$ of $k$ attributes, and a $|S|$-dimensional binary vector $\beta$ specifying their values. The result for this query is a count of the number of people in the database whose attribute vector restricted to $S$ agrees with $\beta$.

Privately releasing approximate answers to a set of $k$-way marginal
queries is one of the most important and well-motivated problems in
differential privacy. Information theoretically, the error  complexity
of marginal queries is well-understood: the per-query additive error
is known to be at least $\Omega(\min\{\sqrt{n},d^{\frac{k}{2}}\})$ and
at most $\tilde{O}(\min\{\sqrt{n}
d^{1/4},d^{\frac{k}{2}}\})$. However, no polynomial time algorithm
with error  complexity as low as the information theoretic upper bound
is  known for small $n$. In this work we present a polynomial time
algorithm that, for any  distribution on marginal queries, achieves
average error at most  $\tilde{O}(\sqrt{n} d^{\frac{\lceil k/2
    \rceil}{4}})$. This error bound is as good as the best known
information theoretic upper bounds for $k=2$. The bound implies that
our mechanisms achieve error $\alpha n$ as long as $n$ is
$\tilde{\Omega}(d^{\frac{\lceil k/2\rceil}{2}}\alpha^{-2})$, which is
an improvement  over previous work on efficiently releasing marginals
when $k$ is small and when error $o(n)$ is  desirable. Using private
boosting we are also able to give nearly matching  worst-case error bounds.

Our algorithms are based on the geometric techniques of Nikolov,  Talwar, and Zhang. The main new ingredients are convex relaxations and careful use of the Frank-Wolfe algorithm for constrained convex  minimization. To design our relaxations, we rely on the Grothendieck inequality from functional analysis.
\end{abstract}

\newpage
\section{Introduction}

A basic task in data analysis is the release of a specified set of statistics of the data. In this work, we address the question of the privacy preserving release of the set of low dimensional marginals of a dataset. These are a ubiquitous and important subclass of queries, constituting {\em contingency tables} in statistics and {\em OLAP cubes} in databases. Official agencies such as the census bureau, the Internal Revenue Service, and the Bureau of Labor Statistics all release certain sets of low dimensional marginals for the data they collect.

In this work, the database will be a collection of the data of $n$ individuals, each characterized by $d$ binary attributes. \junk{We note that a non-binary categorical attribute taking $v$ values can be mapped to a set of $v$ or even $\lceil \log_2 v\rceil$ binary ones.} A $k$-way marginal query is specified by a subset $S$ of $k$ attributes, and a $|S|$-dimensional binary vector $\beta$ specifying their values. The result for this query is a count of the number of people in the database whose attribute vector restricted to $S$ agrees with $\beta$. \junk{These queries are also known as conjunction queries in the case where all $\beta_i = 1$.} In this work, we will be interested in releasing all $k$-way marginals of a database in $(\{-1,1\}^d)^n$, for some small integer $k$. 

In many of the settings mentioned above, the data in question contains individuals' private information, and there are ethical, legal or business reasons to prevent the disclosure of individual information. Differential privacy~\cite{DMNS} is a recent definition that gives a strong privacy guarantee even in the presence of auxiliary information\junk{ that may exist}. It has been the subject of extensive research in the last decade, and will be the definition of privacy in this work. Specifically, we will be working with a variant known as $(\epsilon,\delta)$-differential privacy or approximate differential privacy. We are thus interested in differentially private mechanisms for releasing (estimates of) low dimensional marginals. Our mechanisms will release noisy answers to these queries, and we would like to design computationally efficient mechanisms that add as little noise as possible. In particular, we are interested in achieving error per query $\alpha n$ for $\alpha < 1$ and possibly subconstant when the database size $n$ is not too large. 

This problem of differentially private release of marginals has attracted a lot of interest.  The Gaussian noise mechanism~\cite{DinurN03,DworkN04,DMNS} works in a very general setting and adds noise only  $\tilde{O}(d^{\frac{k}{2}})$ to each of the $O(d^k)$ marginals\footnote{Throughout this introduction, we will ignore the dependence on privacy parameters such as $\epsilon$ and $\delta$, and use the $\tilde{O}$ and $\tilde{\Omega}$ notation to hide factors logarithmic in $d^k$.}, independent of $n$. This implies that the error per query is $\alpha n$ as long as $n = \tilde{\Omega}(d^{k/2} \alpha^{-1})$. Barak et al.~\cite{BarakCDKMT07} showed that these noisy answers can be made consistent with a real database without sacrificing accuracy. In general, this bound is tight: Kasiviswanathan et al.~\cite{shiva2010} show that no differentially private mechanism (even for approximate DP) can add error less than $\Omega(\min(\sqrt{n}, d^{\frac{k}{2}}))$ for constant $k$. 

Starting with the work of Blum Ligett and Roth~\cite{BLR}, a long line of work~\cite{DworkNRRV09,boosting, RothR10, HardtR10,GuptaHRU11,HardtLM12} has shown that private mechanisms with error significantly smaller than that of the Gaussian mechanism exist for small $n$. Specifically, an error bound of about $\tilde{O}(\sqrt{n}{d}^{1/4})$ per query is achievable~\cite{HardtR10,GuptaHRU11,HardtLM12}; this error bound nearly matches the lower bound, and implies error per query at most $\alpha n$ for $n$ as small as $\tilde{\Omega}(d^{1/2}\alpha^{-2})$. However, the known algorithms giving these results have running time that is at least exponential in $d$, which may be restrictive in settings where $d$ is large. Ullman and Vadhan~\cite{UllmanV11} show that, assuming the existence of one-way functions, any private mechanism that generates {\em synthetic data} must have running time  $d^{\omega(1)}$ or have error $\Omega(n)$ for some 2-way marginal query. All algorithms cited above, except for private boosting~\cite{boosting}, do produce synthetic data.

Recent work has shown that significantly faster mechanisms are possible for marginal queries, by using sophisticated learning theory techniques to design approximate but compact representations of databases. In these works, however (see below for a more detailed comparison), either the running time is still $2^{d^{\Omega(1)}}$ for $k = O(1)$, or the error is still much larger than what is achievable inefficiently.

In this work, we show that for any distribution over $k$-way marginals, in time polynomial in $n$ and $d^k$ one can achieve additive error which is %within a $\tilde{O}(d^{\frac{\lceil k/2\rceil}{4} - \frac{1}{2}})$ factor of the best known inefficient mechanism (and within a $\tilde{O}(d^{\frac{\lceil k/2\rceil}{4}})$ factor of the lower bound . Thus for 
 within an $\tilde{O}(d^{\frac{\lceil k/2\rceil}{4}})$ factor of the lower bound. 
\begin{theorem}
For any distribution $\dee$ over $k$-way marginal queries, there is an $(\eps,\delta)$-differentially private mechanism $\mech$ such that for any database containing $n$ individuals
\begin{align*}
\ex_\mech\ex_{q \sim \dee}[\err(q)] = O\left(\frac{\sqrt{n}d^{\frac{\lceil k/2\rceil}{4}} (\log 1/\delta)^{1/4}}{\eps^{1/2}}\right),
\end{align*}
where $\err(q)$ is the additive error incurred by the mechanism on query $q$. \junk{Equivalently, for
\begin{equation*}
  n = \Omega\left(\frac{d^{\frac{\lceil k/2\rceil}{2}} (\log 1/\delta)^{1/2}}{\alpha^2\eps}\right),
\end{equation*}
we have $\ex_\mech \ex_{q \sim \dee}[\err(q)] \leq \alpha n$.}

Moreover the mechanism $\mech$ runs in time polynomial in $d^k$ and $n$.
\end{theorem}
\junk{\ktnote{Do we really get $\sqrt{c(\eps,\delta)}$ and not
  $c(\eps,\delta)$?}
\annote{For average error, definitely. We are basically implementing
  Algo 5, which has this property. Using a relaxation L does not
  change that. For worst case error, once we optimize $s$, I am pretty
  sure the same happens.}}
We note that for $k=2$, this matches the error of the best known (inefficient) mechanism. %Our mechanism is based on the recent geometric approach of Nikolov, Talwar and Zhang~\cite{geometry2013}, who gave a simple mechanism with near optimal (up to a $O(d^{\frac{1}{4}})$ factor) average additive error. The only inefficient step in that work involved a least square projection onto a polytope $K$ defined by the queries. Our main technical contribution is to show that there is a convex body $L$ such that the least squares projection to $L$ can be done efficiently and yet $L$ approximates $K$ well enough for us to to prove good error bounds for the mechanism.

Further, we show that this average error bound can be converted to a worst case bound using the boosting framework of Dwork, Rothblum, and Vadhan~\cite{boosting}. 
\begin{theorem}
Let $2^{-n} \leq \delta \leq n^{-2}$, and $d \leq 2^n$. For any $k$, there is an $(\eps,\delta)$-differentially private mechanism $\mech$ such that for any database containing $n$ individuals, and for all $k$-way marginal queries $q$,
\begin{align*}
\err(q) \leq O\left(\frac{\sqrt{n}  \cdot d^{\frac{\lceil k/2\rceil}{4}} \cdot  (k\log d +\log 1/\delta)^{1/2} (\log n)^{1/4}(k\log d)^{3/4} }{\eps^{1/2}}\right)	.
\end{align*}

Moreover the mechanism $\mech$ runs in time polynomial in $d^k$ and $n$.
\end{theorem}
%\ktnote{Put the correct bounds in these informal theorems, logs and all.}

%This is achieved using results on private boosting (for queries) by Dwork, Rothblum and Vadhan~\cite{boosting}. These results require that a set of answers that give small average error on any distribution can be represented by a small synopsis. We show a point in our relaxation $L$ can be represented using a small number of bits, thus defining such a synopsis.

For $k=2$ our worst-case error bound is in fact an improvement on the error bound achieved with the (inefficient) synopsis generator in~\cite{boosting}; the reason behind our improvement is that we are able to compute more concise synopses. \junk{We note that such an improvement holds in general if we are not concerned with running time: for any set of counting queries, the projection algorithm from~\cite{geometry2013} can be modified to compute synopses concise enough to essentially match the error bounds achievable by private multiplicative weights~\cite{HardtR10} --- further details are omitted from this version of the paper. }

%Some discussion of the connection to k-xor and its consequences on everywhere approximation.

\subsection{Techniques}

Our mechanism is based on the recent geometric approach of Nikolov, Talwar and Zhang~\cite{geometry2013}, who gave a simple mechanism with near optimal (up to a $O(d^{\frac{1}{4}})$ factor) average additive error for $k$-way marginals. This mechanism requires least squares projection onto $nK$, where $n$ is the number of people in the database, and $K$ is the symmetric convex hull of the columns of the query matrix. \junk{In the case of $k$-way parity queries over $d$ attributes, a column corresponding to a vector $x \in \{-1,1\}^d$ has coordinates $v_S = \prod_{i\in S} x_i$, for every $S \subseteq [d]$ such that $|S|=k$.} The Frank-Wolfe algorithm~\cite{frank-wolfe} shows that if one can efficiently optimize linear functions over $K$, then using that as a subroutine, one can compute an approximately optimal projection. This reduces the projection step to optimizing a linear function over $K$. However, optimizing a linear function over $K$  derived from the $k$-way marginal queries generalizes \maxkxor, and is thus $\mathsf{NP}$-hard. Thus one cannot hope to optimize over $K$ in time $\poly(d)$.

%Our main technical contribution is to show that there is a convex body $L$ such that the least squares projection to $L$ can be done efficiently and yet $L$ approximates $K$ well enough for us to to prove good error bounds for the mechanism.
To get around this obstacle we observe that it suffices to project to a polytope $L$ containing $K$ such that $L$ approximates $K$ well enough.  A natural approach is to find an $L$ that we can optimize over, satisfying $K \subseteq L \subseteq C\cdot K$ for some small $C$. For $k=2$, the existence of such a relaxation follows from Grothendieck's inequality~\cite{grothendieck} in functional analysis. Informally, it shows that the maximum value of the quadratic form $\sum_{ij} g_{ij} x_i x_j$ over $x \in \{-1,1\}^d$ is within a constant factor of the maximum value of $\sum_{ij} g_{ij} \langle u_i, u_j \rangle$ over unit vectors $u_i$. The former is closely related to linear functions over the polytope $K$ derived from $2$-way marginals; the latter is a semidefinite program and its feasible region (appropriately projected) gives us the convex body $L$.

%However, for $k\geq 3$, the gap examples for \maxkxor~\cite{Schoenebeck08} (see also Section~\ref{sec:maxkxorgap}) show that natural relaxations $L$ lead to $C$ being too large to be useful. 
However, for $k\geq 3$, the best known relaxations of $K$, following from the work of Khot and Naor~\cite{KhotN07}, lose a factor of $\tilde{\Omega}(d^{\frac{k-2}{2}})$. We are able to do better that the resulting bound: the analysis in~\cite{geometry2013} shows that all one needs is that $L$ contains $K$, and that the  projection of $L$ on a random Gaussian (known as the mean width of $L$) is not much larger than that of $K$. We obtain such a relaxation in two steps: we first relax $K$ to an intermediate polytope $K'$ whose mean width we can bound and which is easier to relax further. Then we show that using the approach above based on Grothendieck's inequality, we can approximate $K'$ by a slightly larger polytope $L$ which we can optimize over.

The above approach gives us average error bounds for any distribution on queries. To get a worst case error bound, we use the Boosting for Queries framework of~\cite{boosting}. This requires that answers returned by the average-error algorithm have a concise representation. We can show that these answers can be represented by a relatively small number of vectors $u_i$ as above. However, {\em a priori} they may be in a high dimensional space. Using the Johnson-Lindenstrauss Lemma, these vectors can be projected down to a small number of dimensions without adding too much additional error, allowing us to get a concise representation as needed.

%Techniques overview

\subsection{Related Work}

The most closely related works to ours are those of Thaler, Ullman and Vadhan~\cite{ThalerUV12}, and  Chandrasekaran et al.~\cite{ChandrasekaranTUW13} and we discuss these in more detail next. Improving on a long line of work~\cite{GuptaHRU11,CheraghchiKKL12,HardtRS12}, the authors in~\cite{ThalerUV12} show that one can construct a private synopsis of a dataset in time $d^{O(\sqrt{k}\log(1/\alpha))}$ such that any $k$-way marginal query can be answered from it, with error $\alpha\cdot n$, as long as $n$ is at least $d^{O(\sqrt{k\log(1/\alpha)})}$.  For a constant $\alpha$, the algorithm in~\cite{ThalerUV12} has the advantage of being online, and when only a few of the $d^k$ queries are asked, has running time much smaller than $d^k$. However, they add error that is much more than necessary (e.g. the best inefficient mechanisms get error $\alpha n$ as long as $n$ is $\Omega(k\sqrt{d}\alpha^{-2})$). 

%\annote{I think the running time of~\cite{ThalerUV12} is also never much worse than $d^{O(k)}$: their polynomials always have degree at most $k$, because at that degree conjunctions can be represented exactly.}

The recent work of Chandrasekaran et al.~\cite{ChandrasekaranTUW13} presents a different point in the trade-off: they show that one can get error $0.01 n$ for $n$ at least  $d^{0.51}$, with running time $$\min\left\{\exp(d^{(1-\Omega(\frac{1}{\sqrt{k}}))}),\exp(d/\log^{0.99} d)\right\}$$ for any sequence of $k$-way marginals. Thus they improve on the $2^{d}$ running time even when $k$ is large. However, the running time of this mechanism is still exponential in $d^c$ for some constant $c$ depending on $k$.

If error $\alpha n$ is desired for small (possibly subconstant) $\alpha$, the lower bound on $n$  and the running time of the above algorithms deteriorate quickly. It is instructive to consider the regime in which the error should be at most $n^{1-\gamma}$ for constant $\gamma$. In order to achieve such small error for $k = O(1)$, the algorithms in~\cite{HardtRS12,ThalerUV12} require databases of size as large as required by the Gaussian noise mechanism\footnote{In fact the algorithm in~\cite{ThalerUV12} is exactly the Gaussian noise mechanism for $\alpha < 2^{-k}$.}, i.e.~$n = \tilde{\Omega}\left(d^{\frac{k}{2(1-\gamma)}}\right)$
By contrast our work gives error $n^{1-\gamma}$ as long as $n$ is $\tilde{\Omega}\left(d^{\frac{\lceil k/2\rceil}{2(1-2\gamma)}}\right)$, a nearly quadratic improvement.
Even for  small constant $\alpha$, the database size lower bound $n = \tilde{\Omega}\left(d^{\frac{\lceil k/2\rceil}{2}}\alpha^{-2}\right)$ required by our algorithm improves significantly on previous work for small $k$. Alternately, translating to additive error, the additive error in~\cite{ThalerUV12} is always at least $\Omega(\min\{n2^{-k},d^{k/2}\})$, compared to the $\tilde{O}(\sqrt{n}d^{\frac{\lceil k/2\rceil}{4}})$ bound that we get.

\junk{\ktnote{Check the above term.}

\annote{I think that TUV's algorithm has error at least $\min\{n2^{-k}, d^{k/2}\}$. The reason is that for $\alpha = 2^{-k}$, TUV's polynomials encode exactly the columsn of the query matrix $A$ and their algorithm reduces to standard Laplace noise with $(\eps, \delta)$-composition.}
\ktnote{I guess that seems reasonable: for alpha being $2^{-k}$, the lower bound on $n$ is $d^k$, at which point, Gaussian noise, or even trivial Laplace noise starts looking more attractive. Adding a sentence.}
}

%Detailed comparison with previous work. Mention that other work is online, ours is not.

%Related work

As mentioned above, Kasiviswanathan et al.~\cite{shiva2010} showed that any differentially private mechanism must incur average case additive error $\Omega(\min(\sqrt{n},\sqrt{d^k}))$. This lower bound comes from privacy considerations alone and makes no computational assumptions. Ullman and Vadhan~\cite{UllmanV11} show that assuming one way functions exist, there is an absolute constant $\alpha^*>0$ such that no polynomial time differentially private algorithm can produce {\em synthetic data} that preserve all 2-way marginals up to error $\alpha^* n$, for a database containing $n=poly(d)$ individuals.  Here synthetic data means that the mechanism computes a new database $D'$ drawn from the original universe and to construct the answer to a query $q$ on $D'$, one computes $q(D')$; e.g. the algorithms in~\cite{BarakCDKMT07,BLR,RothR10, HardtR10,GuptaHRU11,HardtLM12} produce such a synopsis,  (in time exponential in $d$). Our results avoid this lower bound in two ways: the synopses produced by our algorithms are synthetic data from \emph{a larger universe}, and, moreover, for worst case error we use boosting for queries, which aggregates different synopses using medians.

\junk{We leave open the question of whether our bounds can be further improved by an efficient algorithm. Moreover, it is an interesting question if a running time of $d^{o(k)}$ can be achieved when the number of queries asked is a small subset of the $k$-way marginals. We also remark that Hardt, Ligett and McSherry~\cite{HardtLM12} give empirical evaluation of the private multiplicative weights mechanism, and show that for many practical datasets, it can be implemented in practice. It would be interesting to empirically evaluate the mechanisms presented in this work and compare the results to~\cite{HardtLM12}.
}
\section{Preliminaries}
\subsection{Notation}

We denote matrices by upper-case letters, and vectors and scalars by
lower case letters. As standard, we define the $\ell_1$ and $\ell_2$
norms on $\R^m$ respectively as $\|x\|_1 = \sum_{i = 1}^m{|x_i|}$, and
$\|x\|_2 = \sqrt{\sum_{i = 1}^m{|x_i|^2}}$. We use $B_1(n)$ to denote the
$\ell_1$ ball of radius $n$, i.e.~$B_1(n) = \{x: \|x\|_1 \leq n\}$,
and we write $B_1$ for the unit ball $B_1(1)$.

By $x \otimes y$ we denote the tensor product of $x$ and $y$. I.e.~$x
\otimes y \in \R^{[m_1] \times [m_2]}$ when $x \in \R^{m_1}$ and $y \in
\R^{m_2}$, and $(x \otimes y)_{i,j} = x_i y_j$. The notation
$x^{\otimes r}$ stands for the $r$-th tensor power of $x$, i.e.~$x$
tensored with itself $r$ times.

We use the notation $\poly(x_1, \ldots, x_k)$ to denote the set
$O(p(x_1, \ldots,x_k))$ where $p$ is a polynomial in the variables
$x_1, \ldots, x_k$, which tend to infinity jointly.

\subsection{Differential Privacy}

A \emph{database} $D \in \univ^n$ of size $n$ is a multiset of $n$
elements from a universe $\univ$. Each element of the database
represent information about a single individual, and the universe
$\univ$ is the set of all possible types of individuals. Two databases
$D, D'$ are \emph{neighboring} if their symmetric difference $D
\triangle D'$ is at most $1$, i.e.~if they differ in the
presence/absence of a single individual.

We can represent a database $D \in \univ^n$ as a \emph{histogram} as follows:
we enumerate the universe $\univ = \{e_1, \ldots, e_N\}$ in some arbitrary but fixed way; the histogram $x$ associated with a database $D$ is a
vector $x \in \R^N$ such that $x_i$ is equal to the number of occurrences
of $e_i$ in $D$. Two very useful (and closely related) facts about the histogram
representation are that $\|x\|_1 = n$ when $x$ is the histogram of a
size $n$ database, and $\|x - x'\|_1 = D \triangle D'$, where $x$ is
the histogram of $D$, and $x'$ is the histogram of $D'$. In this work, we will work with this histogram representation of the database.

In this paper, we work under the notion of \emph{approximate
differential privacy}. The definition follows.

\begin{definition}[\cite{DMNS,odo}]
  A (randomized) algorithm $\mech$ with input domain $\mathbb{\R}^N$
  and output range $Y$ is \emph{$(\eps, \delta)$-differentially
    private} if for every two neighboring databases $x, x'$, and every
  measurable $S \subseteq Y$, $\mech$ satisfies
  \begin{equation*}
    \Pr[\mech({x}) \in S] \leq e^\eps \Pr[\mech({x}') \in S] + \delta.
  \end{equation*}
\end{definition}
\junk{When $\delta = 0$, we are in the regime of \emph{pure differential privacy}.}

\confoption{}{An important basic property of differential privacy is that the
privacy guarantees degrade smoothly under composition and are not
affected by post-processing.

\begin{lemma}[\cite{DMNS,odo}]
  \label{lm:composition}
  Let $\mech_1$ and $\mech_2$ satisfy $(\eps_1, \delta_1)$- and
  $(\eps_2, \delta_2)$-differential privacy, respectively. Then the
  algorithm which on input $\vec{x}$ outputs the tuple
  $(\mech_1(\vec{x}), \mech_2(\mech_1(\vec{x}), \vec{x}))$ satisfies
  $(\eps_1 + \eps_2, \delta_1 + \delta_2)$-differential privacy.
\end{lemma}

Let us also recall the basic Gaussian noise mechanism and its privacy
guarantee. 
\begin{lemma}[\cite{DinurN03,DworkN04,DMNS}]\label{lm:gauss}
  Let $A = (a_i)_{i = 1}^N$ be a $m \times N$ matrix such that
  $\forall i: \|a_i\|_2^2 \leq \sigma^2$. Furthermore, define $c(\eps,
  \delta) = \frac{1 + \sqrt{2\ln(1/\delta)}}{\eps}$. An algorithm
  which on input a histogram $x \in \R^N$ outputs $Ax + w$, where $w \sim N(0,
  c(\eps, \delta)^2 \sigma^2 )^m$, satisfies $(\eps,
  \delta)$-differential privacy.
\end{lemma}

}
\subsection{Linear Queries and Error Complexity}

A query $q:\univ^* \rightarrow \R$ is \emph{linear} if $q(D) = \sum_{e \in
 D}{q(e)}$. We represent a set $\queries$ of $m$ linear queries as a
\emph{query matrix} $A \in \R^{m \times N}$; associating each query $q
\in \queries$ with a row in $A$ and each universe element $e \in \univ$ with a
column, $A$ is defined by $A_{q,e} = q(e)$. The true answers to all
queries in $\queries$ for a database $D$ with histogram $x$ are given by $y
= Ax$. The \emph{sensitivity} of $\queries$ is defined as $\max_{e \in
  \univ}{|q(e)|}$. 

We measure the error complexity of a mechanism $\mech$ according to
two different measures: average error and worst-case error. The
\emph{mean squared error (MSE)} of a mechanism $\mech$ according to a
distribution $p$ on a set of queries $\queries$ is defined by
\begin{equation*}
  \ex_{\mech} \ex_{q \sim p}{|q(D) - \hat{y}_q|^2},
\end{equation*}
where $\hat{y} = \mech(D, \queries)$. \confoption{Note that we can use Jensen's inequality to translate MSE bounds to bounds on average absolute error.}{Notice that, by Jensen's
inequality, the square root of the MSE according to $p$ is an upper bound on average
absolute error according to $p$
\begin{equation*}
  \ex_{\mech} \ex_{q \sim p}{|q(D) - \hat{y}_q|} \leq
  \sqrt{\ex_{\mech} \ex_{q \sim p}{|q(D) - \hat{y}_q|^2}}. 
\end{equation*}
}
The \emph{worst-case} error of a mechanism $\mech$ on a set of queries
$\queries$ is defined by
\begin{equation*}
  \ex_\mech \max_{q \in \queries}{|q(D) - \hat{y}_q|},
\end{equation*}
where $\hat{y}$ is as before. \confoption{}{For any distribution $p$, if
the worst-case error of $\mech$ is $\lambda$, then the MSE according
to $p$ is at most $\lambda^2$.}

\subsection{Marginals and Parities}
\label{sect:marginals-parities}
% copied previous section; editing it and adding stuff.

In this paper we are concerned with low dimensional marginals, which are a
special case of linear queries.
\junk{A set of $m$ \emph{linear queries} is given by a $m \times N$ \emph{query
  matrix} $A$; the exact answers to the queries on a histogram $x$ are
given by the $m$-dimensional vector $y = Ax$. When $A \in \{0,1\}^{m
  \times N}$, we call the queries \emph{counting queries}.}
%define conjunction queries (and parities)
Let ${[d]\choose k} = \{S \subseteq [d]: |S| = k\}$ denote the set of
subsets of $[d]$ of size $k$. For \emph{$k$-way marginals}, the
universe $\univ$ is the set of $d$-dimensional $\{+1,-1\}$
vectors\footnote{This is a simple notational switch from the usual
  $\{0,1\}$ vectors that helps simplify notation.}, i.e.~$\univ =
\{-1, 1\}^d$. Thus, each person is represented in the database $D$ by
$d$ binary attributes. A \emph{$k$-way marginal query} is specified by
a a set $S$ of $k$ attribute indexes, and a $\beta_{i} \in \{-1,1\}$
for each $i\in S$, and is equal to the number of rows in the database
for which the row vector $b$ restricted to the set of attributes $S$
takes the value given by $\beta$. More formally,
$$
\marg_{(S,\beta)}(D) = \sum_{e \in D}{\bigwedge_{i\in S}{(e_{i}=\beta_i)}} = \sum_{e \in
  D}{\prod_{i\in S} \one({e_{i}} = \beta_i)} . 
$$
%$$More formally, lettting $S=\{s_1,\ldots,s_k\}$ such that $s_1\leq s_2 \leq \ldots \leq s_k$, 
%$$
%\marg^_{(S,\beta)}(D) = \sum_{b \in D}{\bigwedge_{i=1}^k{(b_{s_i}=\beta_i)}} = \sum_{b \in
%  D}{\prod_{i=1}^k \one({b_{s_i}} = \beta_i)} . 
%$$
%The query matrix $A$ has dimensions $2^k{d \choose k} \times 2^d$. Each
%row is associated with a subset $S \subseteq [d]$ of size $k$ and a $\beta \ineach
%column is associated with a binary vector $a \in \{0, 1\}^d$. The
%entry of the matrix in the row corresponding to $S$ and $a$ is equal
%to $\Pi_{i \in S}{a_i}$. 
%
%While we have defined only a subset of all marginals above,  fit in our model as well. We can map an
%attribute vector $a \in \{0, 1\}^d$ to a new vector $a' \in \{0,
%1\}^{2d}$, where $a'_i = a_i$ for $1 \leq i \leq d$, and $a'_i = 1 -
%a_{i - d}$ for $d < i \leq 2d$. Applying this transformation to each
%element of a database $D$ gives a new database $D'$ and computing
%conjunction queries on $D'$ gives answers to non-monotonic conjunction
%queries on $D$ as well. For the rest of the paper we focus on
%monotonic conjunctions. 

It will be convenient to work with a slightly different set of queries that we call parity queries. In the same setting as above, a {\em $k$-wise parity} query is specified by a subset $S$ of $k$ attribute indexes. It is given by
$$
\parity_{S}(D) = \sum_{e \in  D}{\prod_{i\in S} {e_{i}}}. 
$$
\confoption{}{In other words, it is the difference of the number of database elements that have an even number of ones in indexes corresponding to $S$ and the number of those that have an odd number of ones. }

We note that these $k$-wise parities correspond exactly to the
degree-$k$ Fourier coefficients of the histogram $x$. \confoption{}{Barak \etal~\cite{BarakCDKMT07} observed the following useful reduction from marginals to parities. 
\begin{lemma}
For any $S\in {[d]\choose k}$ and $\beta \in \{-1,1\}^S$, there are $\{\frac{-1}{2^k},\frac{+1}{2^k}\}$ coefficients $\alpha_{S,\beta,T}$ for $T \subseteq S$ such that for all $D$,
\begin{align*}
\marg_{(S,\beta)}(D) = \sum_{T \subseteq S} \alpha_{S,\beta,T} \cdot \parity_T(D).
\end{align*}
\end{lemma}}
\confoption{}{
\begin{proof}
Note that both the operators $\marg_{(S,\beta)}$ and $\parity_{S}$ are
linear. Thus it suffices to prove the statement for a database
containing a single element in $\{-1,1\}^d$. Finally observe that:
\begin{align*}
\marg_{(S,\beta)}(\{e\}) &= {\prod_{i\in S} \one({e_{i}} = \beta_i)}\\
&= {\prod_{i\in S} \frac{1}{2}(1 + {e_{i}}\beta_i)}\\
&= \frac{1}{2^k} \sum_{T \subseteq S} (\prod_{i \in T} e_{i}\beta_i)\\
&= \frac{1}{2^k} \sum_{T \subseteq S} (\prod_{i \in T} \beta_i)
\cdot\parity_T(\{e\})
 \end{align*}
\end{proof}
The theorem below follows immediately:}\confoption{In the full version, we show the following reductions.}{}
\begin{theorem}
Suppose that for a database $D$, for all $S \in {[d]\choose k}$, we
have estimates $\hat{y}_S$ satisfying $|\hat{y}_S- \parity_S(D)| \leq
\lambda$. Then we can efficiently construct estimates $z_{S,\beta}$
for all $S\in {[d]\choose k}$, and $\beta \in \{-1,1\}^{S}$, such that
$|z_{S,\beta}-\marg_{(S,\beta)}(D)| \leq \lambda$.
\end{theorem}
\confoption{}{\begin{proof}
We set $z_{(S,\beta)} = \sum_{T \subseteq S}\alpha_{S,\beta,T} \cdot
\hat{y}_T$. Thus, by the triangle inequality,
\begin{align*}
|z_{S,\beta}-\marg_{(S,\beta)}(D)| &= \left|\sum_{T \subseteq S}\alpha_{S,\beta,T}\cdot (\hat{y}_T - \parity_T(D))\right|\\
&\leq \sum_{T \subseteq S}|\alpha_{S,\beta,T}|\cdot |\hat{y}_T - \parity_T(D))|\\
&\leq \sum_{T \subseteq S}\frac{1}{2^k}\cdot \lambda = \lambda
\end{align*}
\end{proof}
}
\confoption{}{It will also be useful to have a version of this result for mean squared error.}
\begin{theorem}
Let $\dee$ be a distribution over $k$-way marginals. Then there exists a distribution $\dee'$ over $k$-wise parities such that the following holds. Given estimates $\hat{y}_S$ such that $\ex_{S \sim \dee'}[|\hat{y}_S - \parity_S(D)|^2] \leq \lambda^2$, we can efficiently construct estimates  $z_{S,\beta}$ such that $\ex_{(S,\beta) \sim \dee}[|z_{S,\beta} - \marg_{(S,\beta)}(D)|^2] \leq  \lambda^2$.
\end{theorem}
\confoption{}{\begin{proof}
We define $\dee'$ as follows: we sample an $(S,\beta)$ in $\dee$ and sample a random $T \subseteq S$. The estimate $z_{S,\beta}$ is simply defined to be $\sum_{T}\alpha_{S,\beta,T} \cdot \hat{y}(T)$. Now for any $(S,\beta)$%The claim follows by an easy application of Cauchy-Schwartz and linearity of expectation.
\begin{align*}
|z_{S,\beta} - \marg_{(S,\beta)}(D)|^2
&= \left|\sum_{T \subseteq S}\alpha_{S,\beta,T}\cdot (\hat{y}_T - \parity_T(D))\right|^2\\
&\leq (\sum_{T \subseteq S}|\alpha_{S,\beta,T}|^2)\cdot(\sum_{T \subseteq S} |\hat{y}_T - \parity_T(D))|^2)\\
&= 2^{-k}\cdot(\sum_{T \subseteq S} |\hat{y}_T - \parity_T(D))|^2),\\
\end{align*}
where the inequality follows by Cauchy Schwartz. Finally observe that when $(S,\beta)$ is drawn according to $\dee$, each of the terms in the summation in the last term is distributed according to $\dee'$. By linearity of expectation, the claim follows.
\end{proof}
}
Thus in the rest of the paper, we will concern ourselves with parity queries. When the database is in its histogram representation, these queries are represented by a matrix $A$ with rows indexed by sets $S \subseteq [d]$ and columns indexed by $e\in \{-1,+1\}^d$, with $a_{S,e} = \parity_S(\{e\})$.

\subsection{Convex Geometry}
%define K and explain intuition

For a convex body $K \subseteq \R^m$, the \emph{polar body} $K^\circ$ is
defined by $K^\circ = \{y: \langle y, x \rangle \leq 1~\forall x \in
K\}$. \junk{The fundamental fact about polar bodies we use is that for any
two convex bodies $K$ and $L$
\begin{equation}\label{eq:conv-duality}
  K \subseteq L \Leftrightarrow L^\circ \subseteq K^\circ.
\end{equation}
In the remainder of this paper, when we claim that a fact follows ``by
convex duality,'' we mean that it is implied by
(\ref{eq:conv-duality}). }%(sasho) not really used afaik

A convex body $K$ is \emph{(centrally) symmetric} if $-K = K$. The \emph{Minkowski
  norm} $\|x\|_K$ induced by a symmetric convex body $K$ is defined as 
$\|x\|_K = \min\{r \in \R: x \in rK\}$. The Minkowski norm induced by
the polar body $K^\circ$ of $K$ is the \emph{dual norm} of $\|x\|_K$
and also has the form 
\begin{equation*}
  \|y  \|_{K^\circ} = \max_{x \in K}{\langle x,  y\rangle}. 
\end{equation*}
The dual norm $\|y\|_{K^\circ}$ is also known as the \emph{width} of
$K$ in the direction of the vector $y$. The \emph{mean (Gaussian) width} of $K$
is the expected width of $K$ in the direction of a random Gaussian and
is denoted $\ell^*(K)$, i.e.~$\ell^*(K) = \E \|g\|_{K^\circ}$, where
$g \sim N(0,1)^m$.

For convex symmetric $K$, the induced norm and the dual
norm satisfy H\"{o}lder's inequality:
\begin{equation}
  \label{eq:holder}
  |\langle x, y \rangle| \leq \|x\|_K \|y\|_{K^\circ}.
\end{equation}

A convex body of primary importance in geometric approaches to
designing differentially private mechanisms for linear queries is the
body $K = A B_1$, where $A$ is the query matrix of family $\queries$
of linear queries. The body $K$ is the symmetric convex hull of all
possible vectors of answers $y$ to the queries $\queries$ for a
 database of size 1.\footnote{The symmetric convex hull of a set of
   points $P$ is the convex hull of $P$ and $-P$.} 
Since the queries are linear, it is easy to see that $nK = AB_1(n)$ is
the symmetric convex hull of all possible vectors of answers $y$ for a database
of size $n$.

\section{The \afive\  Algorithm and Relaxations}

A central tool in the present work is is an algorithm for
answering linear queries, first proposed in~\cite{geometry2013}, which
is simply the well-known Gaussian noise mechanism combined with a
post-processing step. The post-processing, a projection onto $nK$, is
the computationally expensive step of the algorithm. Here, in order to
implement this step efficiently, we modify the algorithm
from~\cite{geometry2013} to project onto a relaxation of $nK$, and we
compute an approximate projection using the Frank-Wolfe convex
minimization algorithm. 

\subsection{Frank-Wolfe}

In this subsection we recall the classical constrained convex
minimization algorithm of Frank and Wolfe~\cite{frank-wolfe}, which
allows us to reduce computing an approximate projection onto a convex
body to solving a small number of linear maximization problems. The
algorithm is presented as Algorithm~\ref{alg:fw}. 

\begin{algorithm}
  \caption{\fw}\label{alg:fw}
  \begin{algorithmic}
    \REQUIRE convex body $F \subseteq \R^m$; point $r
    \in \R^m$; number of iterations $T$  
    
    \STATE Let $q^{(0)} \in F$ be arbitrary.

    \FOR{$t = 1$ \TO $T$}
    \STATE Let 
    %\begin{equation}\label{eq:fw-lp}
    $v^{(t)} = \arg\max_{v \in F}{\langle r - q^{(t-1)}, v \rangle}.$
    %\end{equation}
    \STATE Let 
    %\begin{equation}\label{eq:fw-easy}
      $\alpha^{(t)} = \arg \min_{\alpha \in [0,1]} \|r - \alpha
      q^{(t-1)} - (1-\alpha)v^{(t)}\|_2^2.$ 
    %\end{equation}
    \STATE Set $q^{(t)} = \alpha^{(t)} q^{(t-1)} + (1-\alpha^{(t)})v^{(t)}$.
    \ENDFOR
    \STATE \textbf{Output} $q^{(T)}$.
  \end{algorithmic}
\end{algorithm}

We use the following bound on the convergence rate of the Frank-Wolfe
algorithm. It is a refinement of the original analysis of
Frank and Wolfe, due to Clarkson. 

\begin{theorem}[\cite{frank-wolfe,clarkson}]\label{thm:fw}
  Let $q^* = \arg \min_{q \in F}{\|r - q\|_2^2}$. Then $q^{(T)}$
  computed by $T$ iterations of \fw\ satisfies
  $$
  \|r - q^{(T)}\|_2^2 \leq \|r - q^*\|_2^2 + \frac{4\diam(F)^2}{T + 3}.
  $$
\end{theorem}

While this convergence rate is relatively slow, this will not be an
issue for our application, since privacy forces us to work with noisy
inputs anyways. The expensive step in each iteration is computing
$v^{(t)}$, which requires solving a linear optimization problem over
$F$. Computing $\alpha^{(t)}$ is a quadratic optimization problem in a
single variable, and has a closed form solution.

\junk{
On the other hand, one of the main advantages of using
the Frank-Wolfe algorithm is that its outputs have a \emph{sparse
 representation}, as we make clear in the following lemma.

\begin{lemma}\label{lm:fw-sparsity}
 The output $q^{(T)}$ of \fw$_{T}$ lies in the convex hull of the
 points $q^{(0)}, v^{(1)}, \ldots, v^{(T)} \in F$. 
\end{lemma}
\begin{proof}
 Follows directly from the definition of \fw$_{T}$.
\end{proof}

Let us see what the bounds above imply for implementing
\afive$_\nu$. As we argued in the previous subsection, to optimize the
MSE of \afive$_\nu$, it suffices to set $\nu = O(c(\eps,
\delta)n\sqrt{\log N})$. If we use \fw$_T$ with inputs $F= nK$ and $r
= \tilde{y}$ to compute $\hat{y}= q^{(T)}$ in \afive$_\nu$, we get
\begin{equation*}
  \nu \leq \frac{4 \diam(nK)^2}{m(T+3)} \leq \frac{4n^2}{T}.
\end{equation*}
Therefore, in order to get MSE of $O(c(\eps, \delta) n \sqrt{\log N})$
for \afive$_\nu$, it suffices to compute $\hat{y}$ using \fw$_T$ with
$T$ set to 
$$
T \geq \frac{n}{c(\eps, \delta) \sqrt{\log N}}.
$$
\junk{Combining these observations, we have the following corollary.
\begin{theorem}
  
\end{theorem}}
}

\subsection{Projection onto a Relaxation}

Let us consider a query matrix $A$ which is given only implicitly,
e.g.~the $k$-way parities matrix. More generally, we have the
following definition. 

\begin{definition}
  A $m \times N$ query matrix $A$ for a set of linear queries
  $\queries$ over a universe $\univ$ is \emph{efficiently represented}
  if for each $q \in \queries$, and each $e \in \univ$, $A_{q, e}$ can
  be computed in time $\poly(m, \log N)$. 
\end{definition}

Given an efficiently represented $A$, can we approximate $Ax$ with
additive error close to $\sqrt{n}$ in time $\poly(n,m,\log N)$?  Using
the Frank-Wolfe algorithm, and the geometric methods of Nikolov,
Talwar and Zhang~\cite{geometry2013}, this problem can be reduced to
$\poly(n, \log N, \diam(K))$ calls to a procedure solving the
optimization problem $\arg \max_{v \in K}{\langle u, v \rangle}$,
where $K = AB_1$. While this may be a hard problem to solve,
fortunately, the analysis of the algorithm in~\cite{geometry2013} is
flexible and it is enough to be able to solve the problem for a
relaxation $L$ of $K$. Moreover, we need a relatively weak guarantee
on $L$: it should have mean width comparable with that of $K$, and
diameter that is polynomially bounded. Next we define this
modification of the algorithm and the notion of relaxation that is
useful to us.

\begin{algorithm}
  \caption{\afrelax$_{L}$}\label{alg:relax}
  \begin{algorithmic}
    \REQUIRE \emph{(Public)} efficiently represented query matrix $A =
    (a_i)_{i = 1}^N \in [-1,1]^{m \times N}$; a convex body $L
    \subseteq \R^m$; distribution $p = (p_i)_{i = 1}^m$ on $[m]$;
    number of iterations $T$.

    \REQUIRE \emph{(Private)} database $x \in \R^N$, $\|x\|_1 = n$
    
    \STATE Let $P = \diag(p)$.

    \STATE Let $c(\eps, \delta) = \frac{1+\sqrt{2\ln(1/\delta)}}{\eps}$.
    \STATE Sample $w \sim N(0, c(\eps, \delta)^2 m)^{m}$;
    \STATE  Let $\tilde{y}=P^{1/2}Ax + w$. 
    \STATE  Let $\bar{y}$ be the output of $T$ iterations of \fw\ with
    input the convex body $F = nP^{1/2}L$ and the point $r = \tilde{y}$.
    \STATE \textbf{Output} $\hat{y} = P^{-1/2}\bar{y}$.
  \end{algorithmic}
\end{algorithm}

\begin{definition}\label{defn:eff-relax}
  A convex body $L \subseteq \R^m$ is an \emph{efficient relaxation}
  of the convex body $K = AB_1 \subseteq \R^m$, where $A \in \R^{m
    \times N}$, if $K \subseteq L$, and for any $u \in \R^m$ the
  optimal solution of the maximization problem
    $\arg \max_{v \in L}{\langle u, v \rangle}$
  can be approximated to within $\beta$ \junk{in the $\ell_{\infty}$-norm }in time
  $\poly(\log \frac{1}{\beta}, m, \log \|u\|_\infty, \log N)$.
\end{definition}
Notice that if $L$ is an efficient relaxation of $K$, then $QL$ is an
efficient relaxation of $QK$ for any matrix $Q$ with polynomially
bounded entries.

\confoption{The following theorem follows from the results
  in~\cite{geometry2013}, the convergence bounds of the Frank-Wolfe
  algorithm used to compute $\bar{y}$, and
  Definition~\ref{defn:eff-relax}. The details appear in the full
  version of the paper.}{\junk{The proof of the following theorem follows in a straightforward
manner from Theorem~\ref{thm:lse}, Theorem~\ref{thm:fw}, and
Definition~\ref{defn:eff-relax}.} }
\begin{theorem}
\label{thm:afrelax}
  Let $p$ be a probability distribution on $[m]$ and let $P =
  \diag(p)$. Let $L$ be an efficient relaxation of $K = A
  B_1$\confoption{. }{, and  finally let
  $$
  T = \frac{4n\diam(P^{1/2}L)^2}{c(\eps, \delta) \ell^*(L)}.
  $$}
  Then algorithm \afrelax$_L$
 \confoption{\begin{packed_item}}{ \begin{enumerate}}
  \item satisfies $(\eps, \delta)$-differential
    privacy \label{item:privacy};
  \item can be implemented in time $\poly(m, n, \diam(P^{1/2}L), \log
    N)$; \label{item:time}
  \item outputs a point $\hat{y}$ in $n\cdot\Pi_{\supp(p)} L$, where
    $\Pi_{\supp(p)}$ 
    is a coordinate projection onto the support of $p$;\label{item:point}
  \item has MSE with respect to $p$ at most
    \begin{equation*}
      \E \sum_{i = 1}^m{p_i|y_i - \hat{y}_i|^2} = O\left(c(\eps,
        \delta) n \ell^*(P^{1/2}L)\right);
    \end{equation*} \label{item:mse}
  \confoption{}{\item there exists a constant $C$ s.t.~for any $t > 0$,
  \begin{equation*}
      \Pr [\sum_{i = 1}^m{p_i|y_i - \hat{y}_i|^2} > C\cdot c(\eps,
        \delta) n (\ell^*(P^{1/2}L) + t)] \leq \exp(t^2/4).
    \end{equation*} \label{item:tail}
  }
\confoption{\end{packed_item}}{\end{enumerate}}
\end{theorem}
\begin{proof}
  We first prove claim~\ref{item:privacy}. Since $A \in [-1, 1]^{m \times N}$, and
  $\sum{p_i} = 1$, for any column $a_j$ of $A$ we have
  $\|P^{1/2}a_j\|_2 \leq 1$. Then by Lemma~\ref{lm:gauss}, $\tilde{y}$ is
  $(\eps, \delta)$-differentially private. The output $\hat{y}$ is a
  function only of $\tilde{y}$ and not of the private data $x$, and is
  therefore   $(\eps, \delta)$-differentially private by Lemma~\ref{lm:composition}.

  It is easy to verify that $\tilde{y}$ can be computed in time
  $\poly(m, n, \log N)$ given an efficiently represented $A$. Then
  claim~\ref{item:time} follows since, for an efficient relaxation
  $L$, each step of \fw\ can be implemented in time $\poly(m, \log
  \diam(P^{1/2}L), \log N)$.

  To prove claim~\ref{item:point}, notice that (with the convention
  used in the algorithm that $0^{-1/2} = 0$) $P^{-1/2}P^{1/2}$ is in
  fact the coordinate projection $\Pi_{\supp(p)}$, and $\hat{y} \in
  P^{-1/2}(P^{1/2}K)$.
  
  The central claim is the MSE bound in claim~\ref{item:mse}. The
  proof of this bound follows essentially from~\cite{geometry2013},
  and appears to be a standard method in statistics of analyzing least
  squares estimation. The key observation we make in this work is that
  an efficient relaxation with well-bounded mean width is sufficient
  for the proof to go through. We give the full proof next for
  completeness.  
  
    Since $\hat{y}_i = p^{-1/2}_i \bar{y}_i$ (again using the
  convention $0^{-1/2} = 0$),
  $$
  \sum_{i = 1}^m {p_i |y_i - \hat{y}_i|^2} = \sum_{i=1}^m{|p_i^{1/2}
    y_i - \bar{y}_i|^2} = \|P^{1/2}y - \bar{y}\|_2^2.
  $$
  Therefore it is enough to bound $\mathbb{E}_w \|P^{1/2}y -
  \bar{y}\|_2^2$.

  The bound is based on H\"{o}lder's inequality and the following
  fact:
  \begin{align}
    \|\bar{y} - P^{1/2}y\|_2^2 &= \langle \bar{y} - P^{1/2}y, \bar{y} - P^{1/2}y \rangle\notag\\ 
    &= \langle \bar{y} - P^{1/2}y, \tilde{y} - P^{1/2}y\rangle + \langle
    \bar{y} -  P^{1/2}y, \bar{y} - \tilde{y} \rangle\notag\\
    &\leq 2\langle \bar{y} - P^{1/2}y, \tilde{y} - P^{1/2}y \rangle +
    \nu  \label{eq:magic-ineq},
  \end{align}
  where $\nu = c(\eps, \delta) n  \ell(P^{1/2}L)$. By
  Theorem~\ref{thm:fw}, $\nu$ is an upper bound on how well $\bar{y}$
  approximates the true projection of $\tilde{y}$ onto $nP^{1/2}L$,
  i.e.
  \begin{equation*}
%    \label{eq:fw-bound}
    \|\tilde{y} - \bar{y}\|_2^2 < \min_{\bar{y} \in
      nP^{1/2}L}\|\tilde{y} - \bar{y}\|_2^2 + \nu. 
  \end{equation*}

  The inequality (\ref{eq:magic-ineq}) follows from 
  \begin{align*}
    \langle \bar{y} - P^{1/2}y, \tilde{y} - P^{1/2}y \rangle &= \langle
    \tilde{y}- P^{1/2}y, \tilde{y} - P^{1/2}y\rangle + \langle \tilde{y} - P^{1/2}y,
    \bar{y} - \tilde{y} \rangle\\ 
    &= \|\tilde{y} -  P^{1/2}y\|_2^2 + \langle \tilde{y} - P^{1/2}y,\bar{y} - \tilde{y}  \rangle\\
    &\geq \|\tilde{y} - \bar{y}\|_2^2 - \nu  + \langle \tilde{y} -
     P^{1/2}y,\bar{y} - \tilde{y} \rangle\\ 
    &= \langle \bar{y} - \tilde{y}, \bar{y} - \tilde{y} \rangle - \nu
     + \langle \tilde{y} - P^{1/2}y, \bar{y} -  \tilde{y} \rangle\\
    &= \langle \bar{y} - P^{1/2}y,\bar{y} - \tilde{y}\rangle - \nu .
  \end{align*}
  Inequality (\ref{eq:magic-ineq}), $w = \tilde{y} - P^{1/2}y$, and
  H\"{o}lder's inequality imply 
  \begin{equation*}\label{eq:holder-bound}
    \|\bar{y} - P^{1/2}y\|_2^2 \leq 2\langle \bar{y}- P^{1/2}y,  w \rangle + \nu 
    \leq 2\|\bar{y} - P^{1/2}y\|_{P^{1/2}L}\|w\|_{(P^{1/2}L)^\circ} +\nu \leq
    4n \|w\|_{(P^{1/2}L)^\circ}+ \nu.
  \end{equation*}
  Since $w \sim N(0, c(\eps, \delta)^2)^m$, $\frac{1}{c(\eps,
    \delta)} w \sim N(0,1)^m$. We write
  \begin{equation*}
    \E \sum_{i = 1}^m {p_i |y_i - \hat{y}_i|^2} = \E \|\bar{y} -
    P^{1/2}y\|_2^2 
    \leq  4n \E \|w\|_{(P^{1/2}L)^\circ} + \nu = O\left(c(\eps, \delta)n
    \ell^*(P^{1/2}L)\right).
  \end{equation*}

 Also, observe that the function $\|g\|_{(P^{1/2}L)^\circ}$ is a Lipschitz function of $g$ with Lipschitz constant $\diam(P^{1/2}L) \leq 1$. Indeed 
\begin{align*}
\left| |g|_{(P^{1/2}L)^\circ} - |g'|_{(P^{1/2}L)^\circ} \right| &= \left| \max_{v \in P^{1/2}L} \langle v, g\rangle - \max_{v \in P^{1/2}L} \langle v, g' \rangle \right|\\
&\leq \max_{v \in P^{1/2}L} |\langle v, g - g'\rangle|\\
&\leq \diam(P^{1/2}L) \cdot \|g - g'\|_2. 
\end{align*} 
Thus by Gaussian isoperimetric inequality (see
e.g.~\cite{DubhashiP09}), the tail bound in claim~\ref{item:tail}
follows. 
\end{proof}

In the subsequent section we instantiate this theorem with an
efficient relaxation $L$ of $K = AB_1$, where $A$ is the $k$-wise
parities queries matrix. 
\section{Efficient Relaxation for Marginals via Grothendieck's
  Inequality}

{For convenience we expand the query matrix $A$ for the
  $k$-wise parity queries by adding all $k'$-wise parities for $k'
  \leq k$ and replicating each $k'$-wise parity query row $C_{k'}$
  times, where $C_{k'}$ is a constant depending on $k'$. Formally, we
  substitute $A$ with the matrix $(a_e)_{e \in \{\pm 1\}^d}$ where
  $a_e$ is the column vector $a_e = e^{\otimes k}$. Each row in the new matrix is associated with
  a tuple $s \in [d]^k$, and $A_{s,e} =
  \prod_{i=1}^k{e_{s_i}}$. Clearly all $k$-wise marginals can be
  recovered from the new query matrix, and this transformation only
  affects running time by a constant factor depending on $k$.}

Let us consider the convex body
\begin{equation*}
  L_0 = \conv\{w\otimes z: w,z \in \{\pm 1\}^{d^{k/2}}\},
\end{equation*}
if $k$ is even, or
\begin{equation*}
  L_0 = \conv\{w \otimes z: w \in \{\pm 1\}^{d^{(k-1)/2}}, z \in \{\pm  1\}^{d^{(k+1)/2}}\}, 
\end{equation*}
if $k$ is odd. \junk{For the remainder of this discussion we will assume
that $k$ is even for convenience, but we will state all results in
generality.} It is immediate that $K \subseteq L_0$: 
notice that 
$$
K = \conv\{\pm a_e\} = \conv\{\pm e^{\otimes k}: e \in \{\pm 1\}^d\} = \conv\{\pm
  e^{\otimes k/2} \otimes e^{\otimes k/2}: e \in \{\pm 1\}^d\},
$$
for $k$ even, and similarly for $k$ odd. Since $\pm b^{\otimes k/2} \in
\{\pm 1\}^{d^{k/2}}$ for any $b \in \{\pm 1\}^d$, it follows that $K
\subseteq L_0$. \confoption{The following Lemma is follows from
  standard concentration bounds on Gaussian random variables. The
  proof is given in the full version of the paper.}

\begin{lemma}\label{lm:L0-mean-width}
  For all $k$, $K \subseteq L_0$, and moreover
  %\begin{equation*}
    $\ell^*(P^{1/2}L_0) \leq d^{\lceil k/2 \rceil/2}$ and
    $\diam(P^{1/2}L_0) \leq 1$,
  %\end{equation*}
  for any distribution $p$ on $[d^k]$ and $P = \diag(p)$.
\end{lemma}
\confoption{}{\begin{proof}
  We will prove the theorem for $k$ even; the proof for $k$ odd is
  analogous. 

  By the definition of $\ell^*$ and since a linear function is always
  maximized at an extreme point of a convex set,
  \begin{align*}
    \ell^*(P^{1/2}L_0) = \mathbb{E}_g\  \max_{v \in P^{1/2}L_0}{\langle g, v\rangle}
    &= \mathbb{E}_g\  \max_{w, z \in \{\pm 1\}^{d^{k/2}}}{\langle g,P^{1/2}(w\otimes z)\rangle} 
    \end{align*}
  where expectations are taken over $g \sim N(0, 1)^m$.  Let us fix
  some $w$ and $z$. The vector $w \otimes z$ is a vector in $\{\pm
  1\}^{d^k}$, and therefore, $\|P^{1/2}(w \otimes z)\|_2^2 =
  \sum{p_i} = 1$. By stability of Gaussians, $\langle g,
  P^{1/2}(w\otimes z)\rangle \sim N(0, 1)$. Then, 
  \begin{equation*}
    \max_{w, z \in \{\pm 1\}^{d^{k/2}}}{\langle g,P^{1/2}(w\otimes z)\rangle}
  \end{equation*}
  is the maximum of $2^{d^{k/2}} \times 2^{d^{k/2}} = 2^{2d^{k/2}}$
  standard Gaussian random variables. By standard arguments, it is
  known that the expectation of this maximum is at most $O(\sqrt{\log
    2^{2d^{k/2}}}) = O(d^{k/4})$. %For $k$ odd, the vector $w$ is $d^{(k-1)/2}$-dimensional, while $z$ is $d^{(k+1)/2}$-dimensional. Thus the quantity of interest is the maximum of $2^{d^{(k-1)/2}} \times 2^{d^{(k+1)/2}}$ $N(0,1)$ Gaussians and the claim follows. 
\end{proof}
}
The relaxation $L_0$ is not efficient, as maximizing a linear function
over $L_0$ is $\mathsf{NP}$-hard.\footnote{For example, there is an
  easy reduction from the maximum cut problem even for $k=2$. We omit
  the details.} However, we can view the problem of maximizing a
linear function over $L_0$ as the problem of computing the
$\|\cdot\|_{\infty\mapsto 1}$ norm of an associated matrix and this
norm is well approximated by the optimum of a convex
relaxation~\cite{grothendieck,AlonN04}. This connection, that we
explain next, allows us to relax $L_0$ further to an efficient
relaxation. 

Define the relaxation
\begin{align*}
  L = \{h \in \R^{d^k}: \exists \text{ sequences of unit vectors }
  (u_s)_{s \in [d]^{k/2}}, &(v_t)_{t\in  [d]^{k/2}}\\
  &\text{s.t.~}\forall s, t:
  h_{s\cdot t} = \langle u_s, v_t\rangle\},
\end{align*}
for $k$ even, and
\begin{align*}
  L = \{h \in \R^{d^k}: \exists \text{ sequences of unit vectors }
  (u_s)_{s \in [d]^{(k-1)/2}}, &(v_t)_{t\in  [d]^{(k+1)/2}}\\
  &\text{s.t.~}\forall s, t:
  h_{s\cdot t} = \langle u_s, v_t\rangle\},
\end{align*}
for $k$ odd. Above for two tuples $s = (i_1, \ldots, i_{\lfloor k/2\rfloor})$ and $t
= (j_1, \ldots, j_{\lceil k/2\rceil})$, $s\cdot t$ is their concatenation $(i_1,
\ldots, i_{\lfloor k/2\rfloor}, j_1, \ldots, j_{\lceil k/2\rceil})$. 

\begin{lemma}~\label{lm:L-eff-rel}
  $L$ is an efficient relaxation of $L_0$, and therefore of $K$. 
\end{lemma}
\begin{proof}
  Recall that we associate each coordinate direction in $\R^{d^k}$
  with a tuple $(i_1, \ldots, i_k)$. Assume for the remainder of this
  proof that $k$ is even, the odd case is analogous. Given a point $h
  \in \R^{d^k}$, define the matrix $H \in \R^{d^{k/2} \times d^{k/2}}$
  by $H_{s,t} = h_{s \cdot t}$. Then $L$ is in a one-to-one
  correspondence with the convex set of matrices $H \in \R^{d^{k/2}
    \times d^{k/2}}$ that can be extended to a \emph{positive
    semidefinite} $H' \in \R^{2d^{k/2} \times 2d^{k/2}}$. This
  shows that for any $g \in \R^{d^{k}}$ the maximization problem
  \begin{equation*}
    %\max_{h \in L}{\langle g, h\rangle} = \max_{h \in L}{\langle G,
    %H\rangle} = \max_{(u_s), (v_t)}{\sum_{s,t}G_{s,t}\langle u_s, v_t
    %\rangle}
    \max_{h \in L}{\langle g, h\rangle} = \max_{h \in L}{\tr(G^TH)} = \max_{(u_s), (v_t)}{\sum_{s,t}G_{s,t}\langle u_s, v_t \rangle}
  \end{equation*}
  is a semidefinite program, and therefore can be solved to within
  arbitrary accuracy in polynomial time~\cite{separation}. Above $s$
  and $t$ range over $[d]^{k/2}$, and
  $(u_s)$, $(v_t)$ are sequences of unit vectors in Hilbert space.

  To show that $L_0 \subseteq L$, it is enough to argue that all
  extreme points of $L_0$ are in $L$. Take any $w \otimes z \in L_0$,
  i.e. $w, z \in \{\pm 1\}^{d^{k/2}}$ (for $k$ even, and analogously
  for $k$ odd). Define the unit vectors $(u_s)$ and $(v_t)$ to be just
  the one-dimensional vectors $(w_s), (z_t)$; since $(w \otimes
  z)_{s\cdot t} = w_s z_t$, we have shown the inclusion $w \otimes z
  \in L$.
\end{proof}

Lemma~\ref{lm:L-eff-rel} implies that $L$ can be used in \afrelax$_L$. In
order to give error guarantees for \afrelax$_L$, it would be
enough to show that $\ell^*(P^{1/2}L)$ is not much larger than
$\ell^*(P^{1/2}L_0)$ for any distribution $p$. A much stronger property
--- $L_0 \subseteq L \subseteq C L_0$ for a constant $C$ --- is
implied by Grothendieck's inequality, a classical result in functional
analysis. The following formulation of the inequality is due to
Lindenstrauss and Pelczynski~\cite{absolutelysumming}.

\begin{theorem}[\cite{grothendieck}]\label{thm:grothendieck}
  There exists a constant $C$ such that for any $\ell \times \ell$
  real matrix $M$,
  \begin{equation*}
    \max_{w, z \in \{\pm 1\}^\ell}{w^TMz} \leq C \max_{(u_i)_{i =
        1}^\ell, (v_j)_{j = 1}^\ell}{\sum_{i, j}{M_{ij}\langle u_i,
        v_j \rangle}},
  \end{equation*}
  where the maximum on the right hand side ranges over sequences of
  unit vectors $(u_i)_{i = 1}^\ell, (v_j)_{j = 1}^\ell$ in Hilbert
  space. 
\end{theorem}

The following lemma\confoption{, whose proof is deferred to the full version of
the paper, }{} is an immediate consequence of Theorem~\ref{thm:grothendieck}. 

\begin{lemma}\label{lm:LvsL0}
  There exists a constant $C$ such that for every matrix $Q \in
  \R^{d^{k} \times d^k}$,
  %\begin{equation*}
    $\ell^*(QL) \leq C \ell^*(QL_0)$.
  %\end{equation*}
    Moreover, if $Q = P^{1/2}$ where $P = \diag(p)$  and $p$ is a
    probability distribution on $[d^k]$, $\diam(QL) = \diam(P^{1/2}L)
    \leq 1$. 
\end{lemma}
\annote{talk about diameter in these results: it matters for tail
  bounds and running time}
\confoption{}{\begin{proof}
  Assume again that $k$ is even, and the proof will be analogous when
  $k$ is odd. It is enough to show that for any $g \in \R^{d^k}$,
  $\|g\|_{(QL)^\circ} \leq C \|g\|_{(QL_0)^\circ}$. (In fact by
  duality this establishes the stronger result $L \subseteq C L_0$.)
  We have
  \begin{equation*}
    \|g\|_{(QL_0)^\circ} = \max_{w, z \in \{\pm 1\}^{d^{k/2}}}
    \langle g, Q(w \otimes z)\rangle  
    = \max_{w, z \in \{\pm 1\}^{d^{k/2}}}  \langle Q^Tg, w \otimes z\rangle.
  \end{equation*}
  Define $g' = Q^Tg$, and, as in the proof of Lemma~\ref{lm:L-eff-rel},
  define the $d^{k/2} \times d^{k/2}$ matrix $G'$ by $G'_{s,t} =
  g'_{s\cdot t}$, where $s$ and $t$ range over $[d]^{k/2}$. Then we
  have 
  \begin{equation}\label{eq:L0-dual-norm}
    %\|g\|_{(QL_0)^\circ} = \max_{w, z \in \{\pm 1\}^{d^{k/2}}}
    %\langle G', wz^T\rangle = \max_{w, z \in \{\pm
    %1\}^{d^{k/2}}}{w^TG'z}. 
    \|g\|_{(QL_0)^\circ} = \max_{w, z \in \{\pm 1\}^{d^{k/2}}}
    \tr((G')^Twz^T) = \max_{w, z \in \{\pm 1\}^{d^{k/2}}}
    \tr(zw^TG')= \max_{w, z \in \{\pm  1\}^{d^{k/2}}}{w^TG'z}.  
  \end{equation}

  By an analogous argument, we derive the identity
  \begin{equation}\label{eq:L-dual-norm}
    \|g\|_{(QL)^\circ} =  \max_{(u_s), (v_t)}{\sum_{s,t \in
        [d]^{k/2}}{G'_{s,t}\langle u_s, v_t \rangle}},
  \end{equation}
  where $(u_s)$ and $(v_t)$ are sequences of $d^{k/2}$ unit vectors in
  Hilbert space. The first part of the lemma then follows from (\ref{eq:L0-dual-norm}),
  (\ref{eq:L-dual-norm}), and Theorem~\ref{thm:grothendieck}.

For the diameter bound, we note that for any point in $L$, each entry $h_{s\cdot t}$ is the dot product of two unit vectors and hence bounded in absolute value by $1$. Since $p$ is a distribution, the norm of any point in $P^{1/2}L$ is at most $1$. 
\end{proof}
}
Combining the results above gives our main theorem. 

\begin{theorem}\label{thm:mse-main}
  There exists an $(\eps, \delta)$-differentially private mechanism
  $\mech$ that, given any (public) distribution $p$ on $k$-wise parity
  queries and (private) database $D$, computes answers
  $(\hat{y}_S)_{S: |S|=k} = \mech(p,D)$ with MSE with respect to $p$
  \begin{equation*}
    \sqrt{\mathbb{E}_\mech\E_{S \sim p}{|\parity_S(D) - \hat{y}_S|^2}} \leq C \cdot c(\eps,
    \delta)^{1/2} \sqrt{n} d^{\lceil k/2 \rceil / 4},
  \end{equation*}
for a universal constant $C$. Additionally, for any $t>0$,
  \begin{equation*}
    \Pr[\mathbb{E}_{S \sim p}{|\parity_S(D) - \tilde{y}_S|^2} > C \cdot c(\eps,
    \delta)n (d^{\lceil k/2 \rceil / 2}+t)] \leq exp(-t^2/4)
  \end{equation*}
  Moreover, $\mech$ runs in time $\poly(d^k, n)$.
\end{theorem}
\confoption{}{\begin{proof} The mechanism $\mech$ runs \afrelax$_L$
    with the choice of the number of iterations $T$ as in
    Theorem~\ref{thm:afrelax}. By Theorem~\ref{thm:afrelax} and
    Lemma~\ref{lm:L-eff-rel}, $\mech$ runs in time polynomial in $m =
    d^k$ and $n$. Also by Theorem~\ref{thm:afrelax},
  \begin{equation}\label{eq:err-bound}
    \mathbb{E}_\mech\E_{S \sim p}{|\parity_S(D) - \hat{y}_S|^2} = O\left(c(\eps,
        \delta) n \ell^*(P^{1/2}L)\right).
  \end{equation}
  By Lemmas~\ref{lm:L0-mean-width} and~\ref{lm:LvsL0}, $\ell^*(P^{1/2}L) \leq
  C\ell^*(P^{1/2}L_0) \leq d^{\lceil k/2 \rceil / 2}$. Plugging this
  into (\ref{eq:err-bound}) and taking square roots completes the
  proof of the expected MSE. The tail bound follows analogously.
\end{proof}}

\junk{
By the discussion in Section~\ref{sect:marginals-parities}, this also
implies the following corollary.

\begin{cor}
  There exists an $(\eps, \delta)$-differentially private mechanism
  $\mech$ that, given any (public) distribution $p$ on $k$-way marginal
  queries and (private) database $D$, computes answers $$(\hat{y}_{(S,
  \beta)})_{S: |S|=k, \beta \in \{\pm 1\}^S} = \alg(p,D)$$ with MSE with
  respect to $p$
  \begin{equation*}
    \sqrt{\mathbb{E}_\mech\E_{(S, \beta) \sim p}{|\marg_{(S,\beta)}(D) - \hat{q}(S,
        \beta)|^2}} = O(c(\eps, \delta)^{1/2}\sqrt{n} d^{\lceil  k/2
      \rceil / 4}). 
  \end{equation*}
  Moreover, $\mech$ runs in time $\poly(d^k, n)$.
\end{cor}}

\section{Worst Case Error and Boosting}

At a relatively small cost in error and computational complexity, we
can strengthen the guarantees of Theorem~\ref{thm:mse-main} from MSE
bounds for every query distribution to worst-case error bounds. We do
this via the private boosting framework of Dwork, Rothblum, and
Vadhan~\cite{boosting}. 

\junk{They designed a private boosting algorithm whose
error overhead scales with the size of the smallest (in bits)
representation of the answers of the base sanitizier. Here, we use a modified
\afrelax\ as a base sanitizer and exploit the sparsity properties of
the Frank-Wolfe algorithm together with a rounding technique inspired
by the Johnson-Lindenstrauss transform in order to encode the answers
produced by \afrelax\ in a small number of bits.}

\subsection{The Boosting for Queries Framework}

The boosting for queries framework of Dwork, Rothblum, and Vadhan
assumes black-box access to a \emph{base synopsis generator}: a
private mechanism that, given a set of queries sampled from some
probability distribution and a private database, produces a data
structure (the synopsis) that can be used to answer a strong majority of
the queries with error at most $\lambda$. The boosting algorithm runs
the synopsis generator several times and produces a new data structure
that can be used to answer \emph{all} queries with error $\lambda +
\mu$, where $\mu$ is a term that scales with the \emph{bit size} of
the synopsis produced by the base generator. Next we define a base
generator formally and give  the statement of the main result
from~\cite{boosting}. 

\begin{definition}
  A mechanism $\mech$ is a $(\kappa, \lambda, \beta)$-\emph{base synopsis
    generator} for a set of queries $\queries$, if there exists a
  \emph{reconstruction algorithm} $\mathcal{R}$ such that the
  following holds. For any distribution $p$ on $\queries$, and any
  private database $D$, when $S$ is a multiset of $\kappa$ queries sampled
  independently with replacement from $\queries$, $\hat{D} =
  \mech(S,D)$ satisfies
  %\begin{equation*}
    $\Pr_{\mech, S}[p(\{q: |q(D) - \mathcal{R}(\hat{D}, q)| \geq
    \lambda\}) > 1/3] < \beta$.
  %\end{equation*}
\end{definition}

\begin{theorem}[\cite{boosting}]\label{thm:boosting}
  Let $\queries$ be a set of $|\queries| = m$ linear queries with
  sensitivity 1, and let $T = C \log m$ for a large enough constant
  $C$. There exists a mechanism $\mech$ that, given access to an
  $(\eps_0, \delta_0)$-differentially private $(\kappa, \lambda,
  \beta)$-base synopsis generator $\mech^{\text{base}}$, satisfies
  $(\eps + T\eps_0, T(\kappa\beta + \delta_0))$-differential privacy
  and, for any private database $D$, in time polynomial in $m$ and the
  running time of $\mech^{\text{base}}$, with probability at least $(1-T\beta)$
  outputs answers $(q^*(D))_{q \in \queries}$ such that
  %\begin{equation*}
    $\forall q \in \queries: |q^*(D) - q(D)| \leq \lambda + \mu$,
  %\end{equation*}
  for 
  %\begin{equation*}
    $\mu = O\left(\frac{\sqrt{\kappa}\log^{3/2} m \sqrt{\log{1/\beta}}}{\eps}\right)$.
  %\end{equation*}
\end{theorem}

The term $\mu$ in Theorem~\ref{thm:boosting} is an error overhead due
to the privacy requirements of the boosting algorithm. To minimize
this overhead, we need to make the number $\kappa$ of queries given to
the base generator as small as possible. A generalization result
proved for the uniform distribution in~\cite{DworkNRRV09} and extended
to arbitrary distributions in~\cite{boosting} shows that it is
sufficient to make $\kappa$ only a constant factor larger than the bit
size of synopsis. \confoption{The argument, which is a small variation
  on the one in~\cite{boosting}, is omitted from this extended abstract.}{We reproduce a version of this argument with a slightly
  weaker assumption.}

\begin{lemma}\label{lm:concise}
  Suppose there exists a mechanism $\mech$ and a reconstruction
  algorithm $\mathcal{R}$ such that, given any distribution
  $\tilde{p}$ on the query set $\queries$, and a private database $D$,
  $\hat{D} = \mech(\tilde{p}, D)$ satisfies the MSE bound
  \begin{equation*}
    \sqrt{\sum_{q \in S}{\tilde{p}(q)|q(D) - \mathcal{R}(\hat{D},
        q)|^2}} \leq \lambda, 
  \end{equation*}
  with probability $1-\beta$, and, moreover, for all $D$, $\hat{D}$
  can be represented by a string of $s$ bits. Then $\mech$ is a
  $(O(s+\log 1/\beta), O(\lambda), 2\beta)$-base synopsis generator
  for $\queries$.
\end{lemma}
\confoption{}{\begin{proof}
  Let $p$ be an arbitrary distribution on $\queries$, and let
  $S$ be a multiset of $\kappa$ queries sampled independently with replacement
  from $\queries$. Let $\tilde{p}$ be the \emph{empirical
    distribution} given by $S$, i.e.~$\tilde{p}(q)$ is equal to the
  number of copies of $q$ in $S$ divided by $\kappa$. 
    Define $\queries^{\text{bad}}(\hat{D}, D) = \{q: |q(D) - \mathcal{R}(\hat{D},
    q)| \geq 2\sqrt{3}\lambda\}$. Fix some $\hat{D}$ such that
    $p(\queries^{\text{bad}}(\hat{D}, D)) > 1/3$. Then $\ex_S\ \tilde{p}(\queries^{\text{bad}}(\hat{D},
    D)) > 1/3$, and, by Chernoff's inequality
    \begin{equation*}
      \Pr_S[\tilde{p}(\queries^{\text{bad}}(\hat{D}, D)) \leq 1/6] \leq
      2^{-\kappa/C}
    \end{equation*}
    for a constant $C$. Setting $\kappa = C(s+\log 1/\beta)$, we get that
    $\tilde{p}(\queries^{\text{bad}}(\hat{D}, D)) \leq 1/6$ with
    probability at most $\beta 2^{-s}$; by a union bound over all
    choices of $\hat{D}$, with probability $1-\beta$ over the random
    choice of $S$, $\tilde{p}(\queries^{\text{bad}}(\hat{D}, D)) >
    1/6$ for all $\hat{D}$ such that $p(\queries^{\text{bad}}(\hat{D}, D)) > 1/3$.

    By Markov's inequality and the MSE guarantee of $\mech$, with
    probability $1-\beta$ over the randomness of $\mech$, $\hat{D} =
    \mech(\tilde{p}, D)$ satisfies $\tilde{p}(\{q:|q(D) - \mathcal{R}(\hat{D}, q)| \geq
    2\sqrt{3}\lambda\})\leq 1/6$. Therefore, except with probability
    $\beta$, if $\hat{D} = \mech(S, D)$, then $\tilde{p}(
    \queries^{\text{bad}}(\hat{D}, D)) \leq 1/6$. From this fact and
    the discussion above we conclude that with probability $1- 2\beta$
    (over both the randomness of $\mech$ and the random choice of $S$),
    $p(\queries^{\text{bad}}(\hat{D}, D)) \leq 1/3$ for $\hat{D} =
    \mech(\tilde{p}, D)$.
\end{proof}}

\newcommand{\ess}{\ensuremath{s}}
\newcommand{\dbig}{M}
\newcommand{\dsmall}{M'}
\newcommand{\param}{\chi}

\subsection{Generating a Concise Synopsis}
\confoption{Lemma~\ref{lm:concise} and Theorem~\ref{thm:boosting} together imply
that the additional error $\mu$ incurred by boosting the MSE guarantee
to a worst case guarantee can be made nearly as small as $\sqrt{s}$,
where $s$ is the size of the base synopsis in bits. Following a
relatively standard argument based on the Johnson-Lindenstrauss lemma,
we show that
via a random projection of the vector solution representing the point $\hat{y} \in L$  onto a suitably chosen number of dimensions, one
can obtain the following result. The full argument is omitted from
this extended abstract due to space constraints.}{
Lemma~\ref{lm:concise} and Theorem~\ref{thm:boosting} together imply
that the additional error $\mu$ incurred by boosting the MSE guarantee
to a worst case guarantee can be made nearly as small as $\sqrt{s}$,
where $s$ is the size of the base synopsis in bits. Next we show how
to modify \afrelax\ so that it produces a synopsis small enough to
make this additional error comparable to the MSE bound we have already
proved.

Without modification, \afrelax, called with a distribution
$\tilde{p}$, outputs a point $\hat{y} \in
n\cdot\Pi_{\supp(\tilde{p})}L$. % in $\Re^{d^k}$, which is a convex combination of a small number of points in $L$; i.e. $y = \sum_{i=1}^T \alpha_{i} y_i$, where $y_i \in L$ for all $i$. Further, we have $\alpha_i \geq 0$, and $\sum_i \alpha_i = n$. Thus to represent $y$,
 Thus for any $s,t$, such that $s \cdot t \in \supp(\tilde{p})$,
 $\hat{y}_{s\cdot t} = n\langle u_s,v_t\rangle$, where $u_s,v_t \in
 \Re^{m}$ and $s, t \in [d]^{k/2}$ (in the even case; as before, the
 odd case is analogous). It is a relatively standard fact that these
 SDP vectors can be projected down to about $O({\log m})$ dimensions
 such that each of the $m$ dot products are preserved up to a small
 constant additive error. We will need subconstant error, and will
 thus have to take many more dimensions. We first give a formal
 statement of the guarantee given by the Johnson-Lindenstrauss lemma. 
\begin{lemma}
Let $u$ and $v$ be unit vectors in $\R^{\dbig}$ and let $\Pi$ be a
$\dsmall \times \dbig$ matrix with entries drawn independently from
$N(0, \frac{1}{\dsmall})$, for $\dsmall < \dbig$. Then $\ex[\langle \Pi u,\Pi v\rangle] = \langle u,v \rangle$ and for any $t\in(0,1)$,
\begin{align*}
\Pr[|\langle \Pi u, \Pi v\rangle - \langle u,v \rangle| > 3t] \leq 6\cdot \exp(-\dsmall t^2/6)
\end{align*}
\end{lemma}
\begin{proof}
%Let $P_i$ be a particular row of $P$. Since the distribution of $P_i$ is rotationally invariant, we can assume that $u=e_1$, and let $v = (\alpha,\beta,0,\ldots,0)$ where $\alpha=\langle u,v\rangle$ and $\beta=\sqrt{1-\alpha^2}$. Thus $\langle P_i, u \rangle = X_{i1}$ and $\langle P_i, v \rangle = \alpha X_{i1} + \beta X_{i2}$. Thus the $i$th row of $P$ contributes $\alpha X_{i1}^2 + \beta X_{i1}X_{i2}$ to the dot product. In expectation, the first term is $\alpha Var(X_{i1}) = \alpha/L$, whereas the the second term has expectation zero. This proves the first part of the claim.
By the Johnson-Lindenstrauss Lemma (see e.g.~\cite{DasguptaG03}), for any vector $w$, $\Pr[\|\Pi w\|^2-\|w\|^2 \geq t \|w\|^2] \leq 2\exp(-\dsmall t^2/6)$.  Conditioning on this event for $w\in\{u,v,(u+v)\}$, and observing that $\|u\|^2,\|v\|^2=1$ and $\|u+v\|^2\leq 4$, we write
\begin{align*}
2|\langle \Pi u, \Pi v\rangle - \langle u,v\rangle| &= |(\|\Pi u+\Pi v\|^2 - \|\Pi u\|^2 - \|\Pi v\|^2) - (\|u+v\|^2 - \|u\|^2 - \|v\|^2)|\\
&\leq t\|u+v\|^2 + t \|u\|^2 + t \|v\|^2\\
&\leq 6t
\end{align*}
\end{proof}

In our setting, we wish to preserve $m=O(d^k)$ dot products
approximately, with probability $(1-\beta)$. Suppose that for a parameter $\param$, we set $t= \param d^{\frac{\lceil k/2\rceil}{4}} / \sqrt{n}$, and $\dsmall = 12\cdot (k\log d  + \log 1/\beta)/t^2$. Then, by a union bound, with probability $(1-\beta)$ a random projection $\Pi$ will satisfy simultaneously for all $u_s,v_t$,
\begin{align*}
|\langle \Pi u_s, \Pi v_t\rangle - \langle u_s,v_t \rangle| \leq 6t.
\end{align*}
Also note that this gives us
\begin{align*}
\dsmall = 12\cdot k\log d  (\log 1/\beta)/t^2 = 12 n \cdot (k\log d  + \log 1/\beta)/(\param^2 d^{\frac{\lceil k/2\rceil}{2}})
\end{align*}
We will let our synopsis be defined as the collection of vectors $\{\Pi u_s\},\{\Pi v_t\}$, with each coordinate being represented to $(\log n + \log \dsmall)$ bits of precision. Note that this truncation at the the $(\log n + \log \dsmall)$-th bit adds at most a $1/n$ additive error to the dot product. Also recalling that $\hat{y}_{s\cdot t}$ is $n\langle u_s, v_t\rangle$, we set $\hat{y}'_{s\cdot t}$ to be $n \langle \Pi u_s, \Pi v_t\rangle$. Thus except with probability $\beta$, every pair $s,t$ satisfies
\begin{align*}
|\hat{y}'_{s\cdot t} - \hat{y}_{s\cdot t}| \leq nt = \param \cdot \sqrt{n} d^{\frac{\lceil k/2\rceil}{4}}
\end{align*}
Theorem~\ref{thm:mse-main} then implies the following.}
\begin{lemma}
\label{lm:mse-compressed}
For any $\param>1$, there exists a mechanism $\mech$ and a reconstruction
  algorithm $\mathcal{R}$ such that, given any distribution
  $\tilde{p}$ on $k$-wise parity queries, and a private database $D$,
  $\hat{D} = \mech(\tilde{p}, D)$ satisfies the MSE bound
  \begin{equation*}
    \sqrt{\sum_{q \in S}{\tilde{p}(q)|q(D) - \mathcal{R}(\hat{D},
        q)|^2}} \leq C \cdot  \sqrt{n} \left(c(\eps,\delta)^{1/2}(d^{\frac{\lceil k/2\rceil}{4}}+ \sqrt{8\log (1/\beta)})+ \param\cdot d^{\frac{\lceil k/2\rceil}{4}}\right), 
  \end{equation*}
  with probability $1-\beta$, for some absolute constant $C$. Moreover, for all $D$, $\hat{D}$
  can be represented by a string of $24 n  \cdot d^{\frac{\lceil k/2\rceil}{2}} \cdot (k\log d  + \log (1/\beta))(\log n + \log\log (1/\beta))/\param^2$ bits.  
\end{lemma}

\begin{corollary}
\label{cor:mse-compressed}
Suppose that $\beta \in (\exp(-n),d^{-k/4}n^{-2})$, and $n\leq d^{\frac{\lceil k/2\rceil}{2}}$. 
For any $\param>1$, there exists a mechanism $\mech$ and a reconstruction
  algorithm $\mathcal{R}$ such that, given any distribution
  $\tilde{p}$ on $k$-wise parity queries, and a private database $D$,
  $\hat{D} = \mech(\tilde{p}, D)$ satisfies the MSE bound
  \begin{equation*}
    \sqrt{\sum_{q \in S}{\tilde{p}(q)|q(D) - \mathcal{R}(\hat{D},
        q)|^2}} \leq C \cdot  \sqrt{n} \cdot d^{\frac{\lceil k/2\rceil}{4}}(c(\eps,\delta)^{1/2}+\param), 
  \end{equation*}
  with probability $1-\beta$, for some absolute constant $C$. Moreover, for all $D$, $\hat{D}$
  can be represented by a string of $48 n  \cdot d^{\frac{\lceil k/2\rceil}{2}} \cdot \log (1/\beta) (\log n) /\param^2$ bits.  
\end{corollary}
%\ktnote{Small issue. If beta is really tiny, then the log log $\beta$ term is larger than log of n}
\subsection{Putting it together}
Combining Lemma~\ref{lm:concise} with Corollary~\ref{cor:mse-compressed}, we conclude that if $\beta \in (\exp(-n),d^{-k/4}n^{-2})$, and $n\leq d^{\frac{\lceil k/2\rceil}{2}}$, then for any $\param>1$, there is a mechanism with running time polynomial in $n$ and $d^k$ that is a $(\kappa,\lambda,\beta)$-base synopsis generator, with
\begin{align*}
%\kappa &= 12 n \cdot k\log d  (\log 1/\beta)(\log n + \log\log 1/\beta)\cdot d^{\frac{\lceil k/2\rceil}{2}}.\\
\kappa &= 48 n  \cdot d^{\frac{\lceil k/2\rceil}{2}} \cdot \log (1/\beta) (\log n) /\param^2.\\
%\lambda &= C \cdot c(\eps,\delta)^{1/2} \sqrt{n} (d^{\frac{\lceil k/2\rceil}{4}}+ \sqrt{8\log 1/\beta})
\lambda &= C \cdot  \sqrt{n} \cdot d^{\frac{\lceil k/2\rceil}{4}}(c(\eps,\delta)^{1/2}+\param).
\end{align*}

Plugging this into Theorem~\ref{thm:boosting}, we get our main result
for worst-case error. 
\begin{theorem}\label{thm:main-worstcase}
  Let $2^{-n} \leq \delta \leq n^{-2}$, and $d \leq \exp(n)$. There exists an $(\eps, \delta)$-differentially private
  mechanism $\mech$ that, given any database $D$, with constant probability computes answers
  $(\hat{y}_S)_{S: |S|=k} = \alg(D)$ with worst-case error
  \begin{align*}
  \max_S |\parity_S(D) - \hat{y}_S| = O\left(\frac{\sqrt{n}  \cdot d^{\frac{\lceil k/2\rceil}{4}} \cdot  (k\log d +\log 1/\delta)^{1/2} (\log n)^{1/4}(k\log d)^{3/4} }{\eps^{1/2}}\right).
  \end{align*}
  Moreover, $\mech$ runs in time $\poly(d^k, n)$.
\end{theorem}
\begin{proof}
We set $T=C\log m = Ck\log d$, $\beta = \delta/(\kappa T)$, and use the $(\kappa,\lambda,\beta)$-base synopsis generator with privacy parameters $(\epsilon/T,\delta/T)$. If  $n\geq d^{\frac{\lceil k/2\rceil}{2}}$, then the required error bound is achieved by the Gaussian noise mechanism. Thus we can assume that $n\leq d^{\frac{\lceil k/2\rceil}{2}}$. The error resulting from Theorem~\ref{thm:boosting} is $\lambda+\mu$, where
\begin{align*}
%\lambda &= C \cdot c(\eps,\delta)^{1/2} \sqrt{n} (d^{\frac{\lceil k/2\rceil}{4}}+ \sqrt{8\log 1/\beta})
\lambda &= C_1 \cdot  \sqrt{n} \cdot d^{\frac{\lceil k/2\rceil}{4}}(c(\eps/T,\delta/T)^{1/2}+\param).\\
\mu &= C_2\cdot\frac{\sqrt{\kappa}(k\log d)^{3/2}\sqrt{(\log \kappa + \log T +\log 1/\delta)}}{\epsilon}.\\ 
\end{align*}
Note that $T \leq \kappa$ and $\log \kappa = O(k\log d)$. Also assuming that $d \leq \exp(n)$ and $\delta \geq \exp(-n)$, $\log\log 1/\beta$ is $O(\log n)$. We then get 
\begin{align*}
\mu &= C'_2 \sqrt{n}  \cdot d^{\frac{\lceil k/2\rceil}{4}} \cdot (k\log d +\log 1/\delta) \sqrt{\log n}(k\log d)^{3/2} /(\eps\param).\\
\end{align*}
We can now choose $\param = \left( (k\log d +\log 1/\delta) \sqrt{\log n}(k\log d)^{3/2} /(\eps)\right)^{\frac{1}{2}}$. Observe that $\param > c(\eps/T,\delta/T)^{1/2}$. Thus the $\mu$ term dominates and is equal to
\begin{align*}
\mu &= C'_2 \sqrt{n}  \cdot d^{\frac{\lceil k/2\rceil}{4}} \cdot \param.\\
&= O\left(\frac{\sqrt{n}  \cdot d^{\frac{\lceil k/2\rceil}{4}} \cdot  (k\log d +\log 1/\delta)^{1/2} (\log n)^{1/4}(k\log d)^{3/4} }{\eps^{1/2}}\right).
\end{align*}
The claim follows.
\end{proof}

Observe that the the error bound of
Theorem~\ref{thm:main-worstcase} is, up to logarithmic factors, $\sqrt{n}d^{\frac{\lceil k/2\rceil}{4}}$.

\section{Conclusion}

We have presented our algorithms as mechanisms that satisfy average
error guarantees, similarly to~\cite{geometry2013}; then boosting is
viewed a reduction from worst-case error to average error. An
alternate view is that we attack the problem of achieving worst-case
error bounds for marginals efficiently via designing a new synopsis
generator for the boosting algorithm. Boosting is a natural starting
point for designing efficient mechanisms when the number of queries is
much smaller than the universe size, because, unlike private
multiplicative weights, the running time of the boosting algorithm
does not depend on the universe size but only on the number of queries
and the running time of the synopsis generator. It is an interesting
question whether the geometric techniques
of~\cite{10vollb,12vollb,geometry2013}, which are well-suited to
proving average error upper bounds, can be used to design efficient
synopsis generators for other classes of queries. 

We leave open the question of whether our bounds can be further
improved by an efficient algorithm. Moreover, it is an interesting
question if a running time of $d^{o(k)}$ can be achieved when the
number of queries asked is a small subset of the $k$-way marginals. We
also remark that Hardt, Ligett and McSherry~\cite{HardtLM12} give
empirical evaluation of the private multiplicative weights mechanism,
and show that for many practical datasets, it can be implemented in
practice. It would be interesting to empirically evaluate the
mechanisms presented in this work and compare the results
to~\cite{HardtLM12}.

\section*{Acknowledgements}

The authors would like to thank Moritz Hardt and Anupam Gupta for
thoughtful discussions of the results in this paper. 

%\begin{center}
%\bibliographystyle{alpha}
\bibliographystyle{abbrv}
\bibliography{privacy}
%\appendix
%\input{connection}
\end{document}